\documentclass[%
reprint,
superscriptaddress,
 amsmath,amssymb,
 aps,
pra,
]{revtex4-2}
\usepackage{graphicx}% Include figure files
\usepackage{subfigure}
\usepackage{dcolumn}% Align table columns on decimal point
\usepackage{bm}% bold math
\usepackage{amsmath}
\usepackage{amsthm}
\usepackage{xcolor}
\usepackage{ulem}
\usepackage{pifont}
\usepackage{textcomp}
\usepackage[colorlinks, linkcolor = cyan, anchorcolor = blue, citecolor = magenta]{hyperref}
\usepackage{amsfonts}
\usepackage{dsfont}

\usepackage{graphicx,accents}

\makeatletter
\DeclareRobustCommand{\cev}[1]{%
  \mathpalette\do@cev{#1}%
}
\newcommand{\do@cev}[2]{%
  \fix@cev{#1}{+}%
  \reflectbox{$\m@th#1\vec{\reflectbox{$\fix@cev{#1}{-}\m@th#1#2\fix@cev{#1}{+}$}}$}%
  \fix@cev{#1}{-}%
}
\newcommand{\fix@cev}[2]{%
  \ifx#1\displaystyle
    \mkern#23mu
  \else
    \ifx#1\textstyle
      \mkern#23mu
    \else
      \ifx#1\scriptstyle
        \mkern#22mu
      \else
        \mkern#22mu
      \fi
    \fi
  \fi
}

\makeatother
\newtheorem{thm}{Theorem}

\newtheorem{prop}[thm]{Proposition}
\newtheorem{rem}[thm]{Remark}

\newtheorem{criterion}{Criterion}

\theoremstyle{definition}
\newtheorem{example}{Example}
\newcommand{\bx}{\begin{example}}
\newcommand{\ex}{\end{example}}

\newcommand{\N}{\mathcal{N}}
\newcommand{\M}{\mathcal{M}}
\newcommand{\B}{\mathcal{B}}
\newcommand{\tr}{\text{\rm Tr}}

\newcommand{\wz}[1]{\textcolor{purple}{#1}}

\begin{document}
\preprint{APS/123-QED}

\newcommand{\Q}{\ensuremath{Q^{(1)}}}
\newcommand{\be}{\begin{equation}}
\newcommand{\ee}{\end{equation}}

\title{On small perturbations of coherent information}

\date{\today}% It is always \today, today,

%%%%% author %%%%%%%%%%%%%%%%%%%%%%%%%%%%%%%%%%%%

\author{Zhen Wu}
\affiliation{School of Mathematics and Statistics, Hainan University, Haikou, Hainan Province, 570228, China.}
\affiliation{School of Mathematical Sciences, MOE-LSC, Shanghai Jiao Tong University, Shanghai, 200240, China}
\author{Zhihao Ma}
\affiliation{School of Mathematical Sciences, MOE-LSC, Shanghai Jiao Tong University, Shanghai, 200240, China}
\affiliation{Shanghai Seres Information Technology Co., Ltd, Shanghai 200040, China}
\affiliation{Shenzhen Institute for Quantum Science and Engineering, Southern University of Science and Technology, Shenzhen 518055, China}
%\email{mazhihao@sjtu.edu.cn}
\author{James Fullwood}
\email{fullwood@hainanu.edu.cn}
\affiliation{School of Mathematics and Statistics, Hainan University, Haikou, Hainan Province, 570228, China.}

\begin{abstract}
Quantum capacity quantifies the amount of quantum information that can be transmitted by a quantum channel with an arbitrary small probability of error. Mathematically, the quantum capacity is given by an asymptotic formula involving the one-shot quantum capacity of the associated channel, which, due to purely quantum effects such as superadditivity of one-shot quantum capacity, is rarely computable. The one-shot quantum capacity is mathematically characterized in terms of optimizing an entropic quantity referred to as \emph{coherent information} over all possible input states of a channel, the computation of which also tends to be intractable due to the difficulty of optimizing the coherent information. In this work, we develop perturbative methods for analyzing the behavior of first and second order derivatives of coherent information of a quantum channel with respect to small perturbations of the input state. By doing so, we are able to derive three general criteria for determining whether an input state yields suboptimal coherent information. We then show how our criteria yield sufficient conditions for superadditivity of one-shot quantum capacity, and also for detecting a positive gap between one-shot private capacity and one-shot quantum capacity. The utility of our criteria is illustrated through examples, which yield new results regarding the one-shot quantum capacity of qubit depolarizing channels, Pauli channels and dephrasure channels.   
\end{abstract}

\maketitle

\section{Introduction}
Classical information theory was established by the pioneering work of Shannon, who in a single paper laid the theoretical foundation for nearly all of modern communication technology~\cite{Shannon}. A fundamental result of Shannon was the mathematical characterization of the maximum rate of information that can be transmitted by a classical noisy channel $N$ with an arbitrarily low probability of error, which is referred to as the \emph{channel capacity} of $N$. If $X$ and $Y$ denote the input and output alphabets of a classical channel $N$, and $\boldsymbol{p}(x)$ is a distribution on $X$, then the \emph{mutual information} between sender and receiver is defined to be the non-negative real number $I(X:Y)$ given by
\begin{equation} \label{MTLINFXM67}
I(X:Y) = \sum_{x,y} \boldsymbol{p}(x,y) \log \frac{\boldsymbol{p}(x,y)}{\boldsymbol{p}(x)\boldsymbol{p}(y)}\, ,
\end{equation}
where $\boldsymbol{p}(x,y)=\boldsymbol{p}(x)\boldsymbol{p}(y|x)$ is the associated joint distribution and $\boldsymbol{p}(y)$ is the induced distribution $Y$. What Shannon proved is that the channel capacity of $N$ is the real number $C_{\mathrm{Shan}}(N)$ given by
\begin{equation} \label{CLSCPXTY87}
C_{\mathrm{Shan}}(N) = \max_{\boldsymbol{p}}\ I(X:Y)\,,
\end{equation}
where the maximum is taken over all input distributions $\boldsymbol{p}$ for the channel $N$.  

As we developed methods to better manipulate and control physical systems at the micro-realm, it was realized that subatomic particles---such as photons traveling through an optical fiber---could be used to transmit and process information in a way that exploits characteristically quantum phenomena, such as superposition and entanglement. This then paved the way for the study of quantum information theory~\cite{Wilde,WatrousTQI}, which has established many theoretical results that mark a striking departure from the classical theory of information. In particular, certain quantum communication protocols offer a remarkable advantage over their classical counterparts~\cite{BB84,Simpleproof,unconditionalsecurity,SuperAddIcPauli2}, and the systematic study of such advantages and how to produce new ones provides the theoretical foundation for much of quantum information science.  

Similar to the definition of classical capacity, quantum capacity quantifies the maximum rate of quantum information that can be transmitted through a quantum channel with arbitrarily low probability of error. However, due to quantum effects, the mathematical characterization of quantum capacity is more involved than the classical capacity, as it is given by the following regularized expression~\cite{LSD1,LSD2,LSD3}
\begin{equation}\label{regexp}
    Q(\N) = \lim_{n\rightarrow \infty} \frac{Q^{(1)}(\N^{\otimes n})}{n}\,, \\
\end{equation}
where $\Q(\N)$ is the one-shot quantum capacity of a quantum channel $\N$, which is given by  
\begin{equation}\label{Q1}
    \Q(\N) = \max_{\rho} I_c(\rho,\N)\, ,
\end{equation}
where $I_c(\rho,\N)$ is the \emph{coherent information} of $\N$ with respect to the input state $\rho$. The coherent information $I_c(\rho,\N)$ is a measure of correlation between the output system and the environment after using $\N$ to transmit $\rho$, and is a fundamental measure of quantum information~\cite{CoheInf}.

The formula for the one-shot quantum capacity \eqref{Q1} may be viewed as a quantum analog of formula \eqref{CLSCPXTY87} for the classical capacity, and as such, coherent information is often viewed as a quantum analog of the classical mutual information \eqref{MTLINFXM67}. However, while the classical capacity is additive, i.e., $C_{\mathrm{Shan}}(N_1\times N_2)=C_{\mathrm{Shan}}(N_1)+C_{\mathrm{Shan}}(N_2)$, the one-shot quantum capacity is \emph{superadditive}, since for general quantum channels $\N_1$ and $\N_2$ we instead have the inequality
\begin{equation} \label{superadd}
\Q(\N_1\otimes \N_2)\geq \Q(\N_1)+\Q(\N_2)\, .
\end{equation}
Quantum channels $\N_1$ and $\N_2$ for which the above inequality \eqref{superadd} is strict exhibit holistic properties when combined in parallel which cannot be reduced to properties of its individual constituents, which is a characteristic feature of the transmission of quantum information. In particular, there are many known examples of quantum channels with superadditive one-shot quantum capacity~\cite{SuperAddIcPauli2,IEEEPos,geneticalgorithms,Filippov_2021}, and there even exists channels $\N_1$ and $\N_2$ with $Q(\N_1)=Q(\N_2) = 0$ such that $ Q(\N_1\otimes \N_2)>0$, a phenomenon first discovered by Smith and Yard which is now referred to as \emph{superactivation}~\cite{superactivation,Superact2}. Moreover, for every $n>0$, there exists a quantum channel $\N_{n}$ such that $\Q(\N_{n}^{\otimes m}) = 0$ for all $m \leq n$, yet $Q(\N_{n})>0$~\cite{UnboundedQ}. In addition, superadditivity appears in many information transmission scenarios using quantum channels, indicating that the corresponding channel capacity is non-additive in general. For instance, there are examples showing that superadditivity exists in one-shot private capacity $P^{(1)}$~\cite{StructuredCodes,SuperAddIp1}, and private capacity $P$~\cite{SuperAddP2,SuperAddP}. Furthermore, such superadditive phenomena renders quantum channel capacity to be an elusive quantity which, apart from some special cases, is almost never explicitly computable. Not only this, but even calculating one-shot capacity also requires non-convex optimization over a continuum of states, which is usually an intractable problem. 

Although superadditivity generally holds, it's still possible to find special quantum channels $\B_1$ and $\B_2$ satisfying \emph{additivity}: $Q^{(1)}(\B_1\otimes \B_2) = Q^{(1)}(\B_1) + Q^{(1)}(\B_2)$, which means that we can bypass regularized expressions and obtain the channel capacity by calculating the one-shot capacity of the channel (called single letter formula). For example, in scenarios where quantum channels are used to transmit quantum information, the degradable channel~\cite{LSD2,Structure} has been shown with additive coherent information when combined with another degradable channels, and only entanglement-breaking channels are known with additive coherent information and quantum capacity combined with any other channels~\cite{platypusPRL}.

Recently, Refs.~\cite{Singular,DetectingPositive_NPJ} propose sufficient conditions for determining the positivity of one-shot quantum capacity for generic quantum channels. Given a general quantum channel $\N$, their fundamental idea involves perturbing a pure state to obtain a quantum state $\rho(\varepsilon)$, then considering the derivative of the coherent information for the channel $\N$ with respect to the state $\rho(\varepsilon)$, $I_c'(\varepsilon) := \frac{\mathrm{d} I_c[\N(\rho(\varepsilon))]}{\mathrm{d} \varepsilon}$. Since the coherent information of the general channel with respect to any pure state is zero, if this derivative is greater than zero, it proves that there exists a quantum state $\rho(\varepsilon)$ such that the coherent information $I_c(\rho,\N)$ is positive implying positivity of (one-shot) quantum capacity of the channel $\N$.

In this work, we further develop the perturbation methods employed in Refs.~\cite{Singular,DetectingPositive_NPJ} to derive general criteria for determining whether a quantum state yields suboptimal coherent information. Specifically, given a generic channel $\N$ and a quantum state $\rho$ (which may be mixed), we consider the perturbative expansion $\rho(\varepsilon)$ given in Eq.~\eqref{perstate} below, then compute the first and second orders derivative of the function $f(\varepsilon)$ which is defined as the difference in coherent information between the quantum state before and after perturbation. It's clear that the positivity of $f(\varepsilon)$ implies that there exists a quantum state $\rho(\varepsilon)$ with larger coherent information than $\rho$. Based on this observation, we provide general criteria for determining whether there exists a positive value of $f(\varepsilon)$ according to the positivity of its first and second order derivatives.

As an application of our results, we show how our criteria yield sufficient conditions for detecting superadditivity of one-shot quantum capacity. To illustrate the utility of our criteria, we apply them to qutrit platypus channels as studied in~\cite{Lovasz,Singular,platypusTIT,platypusPRL}, which turn out to yield some new results in such a context. For example, in the case of qutrit platypus channels combined with amplitude damping channels, we provide various input state forms in Eq.~\eqref{pertstateNs}, which realizes superadditive one-shot quantum capacity by leveraging Criterion~\ref{cri2}. 

Finally, we employ our criteria to establish a sufficient condition for there to be a positive gap between the one-shot private capacity $P^{(1)}$ of a quantum channel and its one-shot quantum capacity $\Q$, which never exceeds $P^{(1)}$. Such a result may be viewed as an extension of results of Watanabe in Ref.~\cite{Watanabe}, which also established conditions for detecting a positive gap between one-shot private and quantum capacities for a generic quantum channel. Our results however have a wider range of applicability, since for example they can be applied to Pauli channels and dephrasure channels~\cite{DephrasureChannel}, which do not satisfy the conditions of Watanabe.

\section{Preliminaries}

Let $A$ and $B$ denote finite-dimensional quantum systems with Hilbert spaces $\mathcal{H}_A$ and $\mathcal{H}_B$, respectively. A \emph{quantum channel} from $A$ to $B$ consists of a completely positive, trace-preserving (CPTP) linear map 
\[
\mathcal{N}:\mathcal{L}(\mathcal{H}_A)\longrightarrow \mathcal{L}(\mathcal{H}_B)\, ,
\]
where $\mathcal{L}(\mathcal{H})$ denotes the algebra of linear operators on a Hilbert space $\mathcal{H}$. In order to study the one-shot quantum capacity of a quantum channel $\N$, we make use of a formula for one-shot quantum capacity which employs the notion of a \emph{complementary channel} to $\N$, which we now define. 

Given a quantum channel $\N$, there exists an environment $E$ and an isometry $V:\mathcal{H}_A\to \mathcal{H}_B\otimes \mathcal{H}_E$ such that
\[
\N(\rho) = \tr_E(V \rho V^\dagger)\qquad \,\,\, \forall \rho\in \mathcal{L}(\mathcal{H}_A)\, ,
\]
which is referred to as a \emph{Stinespring representation} of the channel $\mathcal{N}$. The complementary channel to $\N$ is the channel $\N^c:\mathcal{L}(\mathcal{H}_A)\to \mathcal{L}(\mathcal{H}_E)$ obtained by tracing out over $B$, i.e.,
\[
\N^c(\rho) = \tr_B(V \rho V^\dagger)\qquad \,\,\, \forall \rho\in \mathcal{L}(\mathcal{H}_A)\, .
\]
The \emph{coherent information} of $\N$ with respect to an initial state $\rho$ is the real number $I_c(\rho,\N)$ given by
\[
I_c(\rho,\mathcal{N})=S(\N(\rho))-S(\N^c(\rho))\, ,
\]
where $S(\rho) = -\tr(\rho \log \rho)$ denotes von~Neumann entropy. It is straightforward to show that the coherent information is independent of the choice of an environment system $E$ which determines the complimentary channel $\N^c$, and thus is completely determined by the channel $\N$ and the input state $\rho$. The one-shot quantum capacity of $\N$ is then given by the formula
\be \label{Q1FRMXS71}
\Q(\N)=\max_{\rho}I_c(\rho,\N)\, ,
\ee
and if $\rho$ is a state such that $\Q(\N)=I_c(\rho,\N)$, then $\rho$ will be referred to as an \emph{optimal state} with respect to the channel $\N$.

%%%%%%%%%%%%%%%%%%%%%%%%%%%%%%%%%%%%%%%%
\section{Perturbative analysis}
%%%%%%%%%%%%%%%%%%%%%%%%%%%%%%%%%%%%%%%%

As computing the one-shot quantum capacity $\Q$ is generally intractable~\cite{HolevoNPC}, in practice one often selects some special input states which are believed to be nearly optimal for coherent information, yielding an approximation to one-shot quantum capacity for the subsequent information processing task. For example, for a mixed unitary channel $\N(\rho) = \sum_i p_i U_i\rho U_i^\dagger $, it is often the case that one simply chooses the maximally mixed state as its input, which is known to be optimal in certain examples. In this section, we fix a quantum channel $\N$ and an input state $\rho$, and we consider a perturbation of $\rho$ with respect to a small parameter $\varepsilon$. We then analyze the behavior of the associated coherent information under such perturbations, which enable us to derive three criteria for whether the state $\rho$ yields suboptimal coherent information. In particular, given a finite collection of states $\{\rho_i\}_{i=1}^{n}$ such that $\rho_i$ is optimal for a channel $\N_i$, our criteria may be used to determine if the product state $\rho=\rho_1\otimes \cdots \otimes \rho_n$ is \emph{not} in fact optimal for the channel $\N=\N_1\otimes \cdots \otimes \N_n$. 

For this, consider the perturbative expansion
\begin{equation}\label{perstate}
    \rho(\varepsilon)=\rho+\sum_{i=1}^\infty \varepsilon^i A^{(i)} \,,
\end{equation}
where $A^{(i)}$ are traceless Hermitian operators such that $\rho(\varepsilon)$ converges to a valid state for sufficiently small $\varepsilon\in (-r,r)$, 
and let $f(\varepsilon)$ be the real-valued function given by
\begin{equation}\label{FNCXN71}
f(\varepsilon)=I_c(\rho(\varepsilon),\N)-I_c(\rho,\N)\, .
\end{equation}
The criteria formulated in this section are for determining when $f(\varepsilon)$ takes on a positive value, which in turn implies that $\rho$ is not optimal with respect to $\N$. To state our criteria we first set some notation in place. In the case that $\N(\rho)$ is not of full rank, we set $\Pi$ to be the projector onto $\text{Ker}(\N(\rho))$, and in the full rank case we simply set $\Pi$ to be the zero operator. We also use the same convention for $\N^c(\rho)$ (where $\N^c$ is the complementary channel to $\N$), so that $\Pi^c$ denotes the projector onto $\text{Ker}(\N^c(\rho))$ when $\N^c(\rho)$ is not of full rank and when $\N^c(\rho)$ is of full rank we simply take $\Pi^c$ to be the zero operator. Our first criterion is the following:
\begin{criterion}\label{cri1}
Suppose $\N(\rho)$ and $\N^c(\rho)$ are not both of full rank, and suppose there exists a perturbative expansion $\rho(\varepsilon)$ of $\rho$ as in \eqref{perstate} such that  
\begin{equation} \label{Cx1Xs}
    \tr \left(\Pi\, \N(A^{(1)})\right) > \emph{Tr} \left(\Pi^c\, \N^c(A^{(1)})\right) \,.
\end{equation}
Then there exists a positive value of $f(\varepsilon)$ as given by \eqref{FNCXN71}.
\end{criterion}

Criterion~\ref{cri1} may be viewed as a generalization of a result which first appeared in Ref.~\cite{DetectingPositive_NPJ}, where perturbative expansions of pure states are used to derive a sufficient condition for the positivity of coherent information for general quantum channels. Such a result was then applied to a qubit depolarizing channel $\N$, where it was shown that the complementary channel $\N^c$ has positive quantum capacity $Q(\N^c) > 0$, a result which was first established in Ref.~\cite{complementaryQuantum}. Criterion~\ref{cri1} however makes no purity assumption on $\rho$, and applies to arbitrary states. This more general setting also allows one to use Criteria~\ref{cri1} for detecting phenomena such as superadditivity of one-shot quantum capacity, as we show later on.

For the statement of the second and third criterion, we need to introduce some further notation: Let $\{\lambda_i\}$ denote the set of distinct eigenvalues of $\N(\rho)$, let $\{\mu_j\}$ denote the set of distinct eigenvalues of $\N^c(\rho)$, let $P_i$ denote the projector onto the $\lambda_i$-eigenspace of $\N(\rho)$ for all $i$, and let $Q_j$ denote the projector onto the $\mu_j$-eigenspace of $\N^c(\rho)$. Moreover, we let $W_\M$ and $W_\M^{(0)}$ denote operators which are determined by the spectral decompositions of $\M(\rho)$ for $\M = \N$ or $\N^c$, whose explicit expressions are given in~\eqref{W0}\eqref{W} of the Supplementary Material~\cite{SuppMaterial}.
\begin{criterion}\label{cri2}
    Suppose $\N(\rho)$ and $\N^c(\rho)$ are not both of full rank, and suppose there exists a perturbative expansion $\rho(\varepsilon)$ of $\rho$ as in \eqref{perstate} such that for all $P_i$ and $Q_j$
\begin{equation}\label{CX1Xs}
    \tr \left(P_i\, \N(A^{(1)})\right) = \emph{Tr} \left(Q_j\, \N^c(A^{(1)})\right) = 0 \,,
\end{equation}
and 
\begin{equation}\label{CX2Xs}
    \tr \left(\Pi \N(A^{(2)}) - W^{(0)}_{\N}\right) > \tr \left(\Pi^c \N^c(A^{(2)}) - W^{(0)}_{\N^c}\right)\, .
\end{equation}
% where $W^{(0)}_{\N}$ and $W_{\N^c}^{(0)}$ are defined by Eq.~\eqref{W1OPRXS87}. 
Then there exists a positive value of $f(\varepsilon)$ as given by \eqref{FNCXN71}.
\end{criterion}

The third criterion covers the case when $\N(\rho)$ and $\N^c(\rho)$ are both of full rank. Since whether these two operators are full rank depends on the dimension of the corresponding space, we must emphasize that the dimensions of output $d^*_{\text{out}}(\N)$ and environment $d^*_{\text{env}}$ system of the channel $\N$ are both minimal, in the sense that $d^*_{\text{out}}(\N) = \mathrm{rank}[\N(\mathds{1})]$ and $d^*_{\text{env}} = \mathrm{rank}[\N^c(\mathds{1})]$ as shown in Lemma II.3 of~\cite{DetectingPositive_NPJ}, and hence the ranks of $\N(\rho)$ and $\N^c(\rho)$ satisfy $\mathrm{rank}[\N(\rho)] = \mathrm{rank}[\N(\mathds{1})]$ and $\mathrm{rank}[\N^c(\rho)] = \mathrm{rank}[\N^c(\mathds{1})]$.
\begin{criterion}\label{cri3}
Suppose the operators $\N(\rho)$ and $\N^c(\rho)$ are both of full rank, and suppose there exists a perturbative expansion $\rho(\varepsilon)$ of $\rho$ as in \eqref{perstate} such that
\begin{equation} \label{CX1X}
\tr\big(\N(A^{(1)})\log\N(\rho)\big) = \tr\big(\N^c(A^{(1)})\log\N^c(\rho)\big)
\end{equation}
and
\begin{equation} \label{CX2X}
\tr\big(W_{\N^c}\big) > \tr\big(W_{\N}\big)\, .
\end{equation}
Then there exists a positive value of $f(\varepsilon)$ as given by \eqref{FNCXN71}.
\end{criterion}

When $\N(\rho)$ and $\N^c(\rho)$ are both of full rank, Ref.~\cite{CohInf_JMP} computes the gradient and Hessian of ($n$-shot) coherent information of the strictly positive channel $\N$ at a positive definite state $\rho$. In contrast, our results propose explicit criteria for determining whether a quantum state is optimal for coherent information.

The proofs of these criteria are rather technical, and are thus given in the Supplementary Material~\cite{SuppMaterial}. The basic idea behind the criteria however is quite straightforward. Since the function $f(\varepsilon)$ satisfies $f(0)=0$ and its derivatives are continuous in a neighborhood of $\epsilon=0$, we are guaranteed a positive value of $f(\varepsilon)$ if either $f'(0)>0$, or $f'(0)=0$ and $f''(0)>0$. When $\N(\rho)$ and $\N^c(\rho)$ are not both of full rank, we show condition \eqref{Cx1Xs} implies $f'(0)>0$, thus giving us our result, while the condition~\eqref{CX1Xs} implies $f'(0) = 0$ and~\eqref{CX2Xs} implies $f''(0) > 0$. On the other hand, when $\N(\rho)$ and $\N^c(\rho)$ are both of full rank, we show condition \eqref{CX1X} implies $f'(0)=0$ and condition \eqref{CX2X} implies $f''(0)>0$, implying a positive value of $f(\varepsilon)$.

\begin{rem}
If we reverse the sense of the inequalities~\eqref{Cx1Xs}\eqref{CX2Xs}\eqref{CX2X} in these criteria, then we obtain criteria for $f(\varepsilon)$ to have a negative value.
\end{rem}

%%%%%%%%%%%%%%%%%%%%%%%%%%%%%%%%%%%%%
\section{Detecting superadditivity of one-shot quantum capacity}
%%%%%%%%%%%%%%%%%%%%%%%%%%%%%%%%%%%%%

In this section, we show how our criteria may be used to detect superadditivity of one-shot quantum capacity, and we provide two examples illustrating the utility of our results. 

\begin{prop}[\wz{\cite{Singular}}] \label{SPRADXTY87}
Let $\rho_1$ and $\rho_2$ be optimal states for quantum channels $\N_1$ and $\N_2$, let $\rho(\varepsilon)$ be a perturbative expansion of $\rho=\rho_1\otimes \rho_2$ as in \eqref{perstate}, and let $f(\varepsilon)$ be the function given by \eqref{FNCXN71} with $\N=\N_1\otimes \N_2$. Then a positive value of $f(\varepsilon)$ implies that the one-shot quantum capacity of $\N_1$ and $\N_2$ is superadditive, i.e., $\Q(\N_1\otimes \N_2)>\Q(\N_1)+\Q(\N_2)$.  \end{prop}

\begin{proof}
Since $\rho_i$ is optimal for $\N_i$, the additivity of coherent information yields
\begin{align*}
I_c(\rho_1\otimes \rho_2,\N_1\otimes \N_2)& =I_c(\rho_1,\N_2)+I_c(\rho_2,\N_2) \\
&=\Q(\N_1)+\Q(\N_2)\, .
\end{align*}
Therefore, if $f(\varepsilon)>0$, we then have
\begin{align*}
\Q(\N_1)+\Q(\N_2)&=I_c(\rho_1\otimes \rho_2,\N_1\otimes \N_2) \\
&<I_c(\rho(\varepsilon),\N_1\otimes \N_2) \\
&\leq \Q(\N_1\otimes \N_2)\, ,
\end{align*}
as desired.
\end{proof}

\begin{rem}
Proposition~\ref{SPRADXTY87} naturally extends to the multi-channel case: Let $\{\M_i\}_{i=1}^n$ be a finite collection of channels, and suppose that $\{\sigma_i\}_{i=1}^n$ is a collection of states such that $\sigma_i$ is optimal for $\M_i$ for all $i$, let $\rho(\varepsilon)$ be a perturbative expansion of $\rho =\sigma_1\otimes\cdots\otimes\sigma_n$ as in \eqref{perstate}, and let $f(\varepsilon)$ be the function given by \eqref{FNCXN71} with $\N = \M_1\otimes\cdots\otimes\M_n$. Then a positive value of $f(\varepsilon)$ implies that the one-shot quantum capacity of $\{\M_i\}_{i=1}^{n}$ is superadditive, i.e.,
\[
\Q\Big(\bigotimes_{i=1}^{n}\M_i\Big)>\sum_{i=1}^{n}\Q(\M_i)\, .
\]
\end{rem}

In Refs.~\cite{platypusPRL,Singular}, it was shown that the qutrit platypus channel has superadditive one-shot quantum capacity when combined with depolarizing channels and amplitude damping channels. In the next example, we show how Criterion~\ref{cri2} may be utilized to recover the result of the latter case, namely, when a qutrit platypus channel is combined with an amplitude damping channel.   

\bx
The qutrit platypus channel $\mathcal{N}_s$~\cite{Lovasz,Singular} is the channel associated with the Stinespring representation determined by the isometry $F_s : \mathcal{H}_A \rightarrow \mathcal{H}_B\otimes \mathcal{H}_E$ given by
\begin{equation*}
    \begin{aligned}
        F_s |0\rangle & = \sqrt{s}|0\rangle\otimes |0\rangle + \sqrt{1-s}|1\rangle\otimes |1\rangle\, , \\
        F_s |1\rangle & = |2\rangle\otimes|0\rangle\, ,\\
        F_s |2\rangle & = |2\rangle\otimes|1\rangle\, ,
    \end{aligned}
\end{equation*}
with $0<s\leq 1/2$. The qutrit platypus channel has been studied extensively in Refs.~\cite{platypusTIT,platypusPRL}, where it was shown that a state of the form $\rho(p) = (1-p)|0\rangle\langle0| + p|2\rangle\langle 2|$ is optimal with respect to $\N_{s}$, so that
\begin{equation*}
    \Q(\mathcal{N}_s) = \max_{p\in [0,1]} I_c\big(\rho(p),\mathcal{N}_s\big)\, .
\end{equation*}
The qubit amplitude damping channel $\mathcal{A}_\gamma$ with damping probability $\gamma \in [0,1]$ admits the Stinespring representation corresponding to the unitary $V_\gamma = A_0\otimes |0\rangle + A_1\otimes |1\rangle$, with $A_0 = |0\rangle\langle 0| + \sqrt{1-\gamma} |1\rangle\langle 1|$ and $A_1 = \sqrt{\gamma} |0\rangle\langle 1|$. It is then straightforward to show that the associated complementary channel is also an amplitude damping channel, with damping probability $1-\gamma$. The quantum capacity of an amplitude damping channel coincides with its one-shot quantum capacity, i.e., $Q(\mathcal{A}_\gamma) = \Q(\mathcal{A}_\gamma)$~\cite{LSD2} for all $\gamma\in [0,1]$, and its quantum capacity is vanishing for $\gamma \geq 1/2$. Moreover, it is known that a state of the form $\rho(u) = (1-u)|0\rangle\langle0| + u|1\rangle\langle 1|$ is optimal for all $\gamma$~\cite{ADC}, so that
\begin{equation*}
    \Q(\mathcal{A}_\gamma) = \max_{u\in [0,1]} I_c\big(\rho(u), \mathcal{A}_\gamma\big)\,.
\end{equation*}

\begin{figure}
    \centering
    \includegraphics[width=1\linewidth]{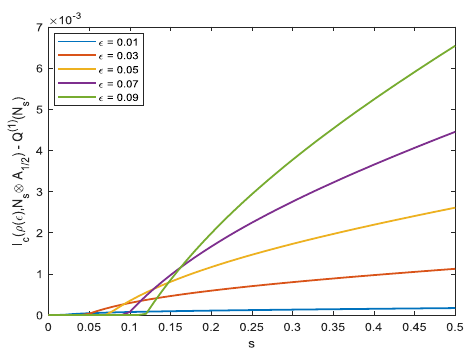}
    \caption{Super-additivity of one-shot quantum capacity for the qutrit platypus channel $\N_s$ and the qubit amplitude damping channel $\mathcal{A}_{1/2}$. Here $\mathcal{A}_{1/2}$ is anti-degradable with zero quantum capacity, and the perturbative expansion $\rho(\varepsilon)$ is given in \eqref{pertstateNs} with $a^2 = 1-\varepsilon^2$.}
    \label{fig1}
\end{figure}

So now let $\rho_1 = (1-w)|0\rangle\langle0| + w|2\rangle\langle 2|$ and $\rho_2 = |0\rangle\langle 0|$ be optimal states with respect to $\N_s$ and $\mathcal{A}_\gamma$ with $\gamma \geq 1/2$, respectively, and let $\rho=\rho_1\otimes \rho_2$. We then consider the perturbative expansion $\rho(\varepsilon)$ given by 
\begin{equation}\label{pertstateNs}
    \rho(\varepsilon) = \rho_1\otimes\rho_2 + \varepsilon A^{(1)} + \varepsilon^2 A^{(2)}\, ,
\end{equation}
where $A^{(1)} = wa (|11\rangle\langle 20| + |20\rangle\langle 11|)$ and $
A^{(2)} = w (|11\rangle\langle 11|-|20\rangle\langle20|)$
with $a^2 \leq 1-\varepsilon^2$ (which ensures $\rho(\varepsilon)$ is a state). 

As the output states $\N_1(\rho_1)\otimes \N_2(\rho_2)$ and $\N_1^c(\rho_1)\otimes \N_2^c(\rho_2)$ are both diagonal which implies that the projectors $P_i$ and $Q_j$ in condition~\eqref{CX1Xs} of Criterion~\ref{cri2} are corresponding to the computational basis, and the diagonal entry of operators $\N_1\otimes \N_2(A^{(1)})$ and $\N_1\otimes\N_2(A^{(1)})$ are all zero, one can know that the operator $A^{(1)}$ given here satisfies condition~\eqref{CX1Xs} of Criterion~\ref{cri2}, we now show the condition~\eqref{CX2Xs} also holds: Let $\N = \mathcal{N}_s\otimes \mathcal{A}_{\gamma}$, so that its complement is $\N^c= \mathcal{N}_s^c\otimes \mathcal{A}_{\gamma}^c$. We then have
\begin{equation*}
    \begin{aligned}
        \tr \left(\Pi \N (A^{(2)}) - W^{(0)}_{\N}\right) & = w(1-\gamma) \,, \,\, \text{and} \\
        \tr \left(\Pi^c \N^c(A^{(2)}) - W^{(0)}_{\N^c}\right) & = w\gamma - \frac{w^2a^2\gamma}{(1-s)(1-w)+w}\,.
    \end{aligned}
\end{equation*}
Since $a^2$ can be chosen to close to $1$ for $\varepsilon$ close to $0$, it follows that condition~\eqref{CX2Xs} of Criterion~\ref{cri2} holds for all $\gamma \in [\frac{1}{2}, \frac{1-s+sw}{2-2s-w+2sw})$, which is the same range of $\gamma$ provided in~\cite{Singular,platypusPRL}, together with Proposition~\ref{SPRADXTY87} implies that the one-shot quantum capacity of $\N_s$ and $\mathcal{A}_\gamma$ is superadditive. We note that the perturbative expansion $\rho(\varepsilon)$ given by \eqref{pertstateNs} is different from the perturbative state $\sigma(\varepsilon)$ utilized in Refs.~\cite{Singular,platypusPRL} to establish superadditivity of $\Q$, which is given by 
\[
\sigma(\varepsilon) = r_1|00\rangle\langle 00| +r_2 |01\rangle\langle 01| + (1-r_1-r_2)|\psi\rangle\langle\psi|\, ,
\]
where $r_1,r_2$ are parameters to be optimized and $|\psi\rangle = \sqrt{1-\varepsilon}|20\rangle + \sqrt{\varepsilon}|11\rangle$. As shown in Fig.~\ref{fig1}, for $\varepsilon = 0.01$, $\Q$ is superadditive for $s>0.01$, while greater $\varepsilon$ gives rise to higher superadditive  effects for $s$ sufficiently large. $\hfill\triangle$
\ex

%\wz{Although we haven't yet found a quantum channel that satisfies Criterion~\ref{cri3} and exhibits the superadditivity of quantum capacity, Criterion~\ref{cri3} still provides a possibility for detection on the superadditivity of quantum capacity.}

As for the applicability of Criterion~\ref{cri3}, we have yet to find any examples of pairs $(\rho,\mathcal{N})$ for which the conditions \eqref{CX1X} and \eqref{CX2X} of Criterion~\ref{cri3} are both satisfied. It could be the case that such conditions are satisfied only on a set of measure-zero, in which case the utility of Criterion~\ref{cri3} would be quite limited. As such, it would be interesting to determine the locus of pairs $(\rho,\mathcal{N})$ satisfying the conditions \eqref{CX1X} and \eqref{CX2X} for fixed input and output systems of the channel $\mathcal{N}$.

%%%%%%%%%%%%%%%%%%%%%%%%%%%%%%%%%%%%%
\section{One-shot quantum capacity vs. one-shot private capacity}
%%%%%%%%%%%%%%%%%%%%%%%%%%%%%%%%%%%%%
 
The one-shot private capacity is the highest rate at which a quantum channel can faithfully transmit classical information without leaking information to an external system, and plays a fundamental role in quantum cryptography~\cite{BB84,E91,QuanCryp}. For its mathematical definition, consider an ensemble $\{p_i, \rho_i\}$ of input states for a quantum channel $\N$, and let the state $\rho = \sum_i p_i\rho_i$ be the associated mixture for this ensemble. Then the private information associated with the ensemble $\{p_i, \rho_i\}$ and channel $\N$ is the real number $I_p(\{p_i, \rho_i\},\N)$ given by 
\begin{equation}\label{privateinformation}
    I_p\left(\{p_i, \rho_i\},\N\right) = I_c(\rho,\N) - \sum_i p_i I_c(\rho_i,\N)\,, 
\end{equation}
and the one-shot private capacity $\mathcal{P}^{(1)}(\N)$ is the maximum of $I_p(\{p_i,\rho_i\},\N)$ over all ensembles, i.e.,
\[
\mathcal{P}^{(1)}(\N)=\max_{\{p_i,\rho_i\}}I_p(\{p_i,\rho_i\},\N)\, .
\]

Given a state $\rho$, let $E_{\rho}$ denote the set of all ensemble decompositions of $\rho$. Since the coherent information with respect to a pure state is necessarily zero, 
\[
\max_{\{p_i,\rho_i\}\in E_{\rho}} I_p\Big(\{p_i, \rho_i\},\N\Big)\geq I_c(\rho,\N)\, ,
\]
since one can always take the pure state decomposition $\{p_i, |\psi_i\rangle\langle \psi_i|\}$ of $\rho$ so that $\rho = \sum_i p_i |\psi_i\rangle\langle\psi_i|$, in which case we have $I_p\big(\{p_i, |\psi_i\rangle\langle \psi_i|\},\N\big) = I_c(\rho,\N)$. 
We then have
\begin{equation*}
\begin{aligned}
    P^{(1)}(\N) & = \max_{\rho} \max_{\{p_i,\rho_i\}\in E_{\rho}} I_p\Big(\{p_i, \rho_i\},\N\Big) \\
                & \geq \max_{\rho} I_c(\rho,\N) = \Q(\N) \, ,
\end{aligned}
\end{equation*}
so that one-shot quantum capacity never exceeds one-shot private capacity. Similar to quantum capacity, the private capacity of a quantum channel $\N$ is then given by the regularized formula
\begin{equation*}
    P(\N) = \lim_{n\rightarrow\infty} \frac{P^{(1)}(\N^{\otimes n})}{n}\,,
\end{equation*}
and is thus rarely computable due not only to superadditive effects, but also due to intractable optimization procedures needed for the calculation of $\mathcal{P}^{(1)}$. However, it is still possible to detect a positive gap between $P^{(1)}$ and $\Q$ without explicitly computing their exact values.
While Watanabe~\cite{Watanabe} proved that $\Q(\B) = P^{(1)}(\B)$ if the one-shot quantum capacity of the complementary channel $\B^c$ vanishes~\cite{Hirche2022,IEEEPartialOrder}. There exist quantum channels displaying the positive gap between one-shot quantum and private capacity, such as Horodecki channels~\cite{Horodeckichannel2,Horodeckichannel3,HorodeckiChannel1,Horodeckichannel4}, half-rocket channels~\cite{Leung2014}, platypus channels~\cite{Lovasz,SDPBoundC,platypusTIT,wu2025} and dephrasure channels~\cite{DephrasureChannel}.
In particular, we now show how our criteria implies a sufficient condition for $P^{(1)}$ to be strictly larger than $\Q$ which is more general than the result of Watanabe~\cite{Watanabe}:

\begin{prop}\label{prop2}
    Suppose $\sigma$ is an optimal mixed state for a quantum channel $\N$, and suppose $\rho(\varepsilon)$ is a perturbative expansion of a pure state $\rho$ such that $\sigma - r\cdot\rho(\varepsilon)$ is positive semi-definite for some $r>0$. Then a negative value of $f(\varepsilon)$ as given by \eqref{FNCXN71} implies $P^{(1)}(\N) > \Q(\N)$.
\end{prop}

\begin{proof}
Since by Eq.~\eqref{privateinformation} we have
\begin{equation*}
\begin{aligned}
    P^{(1)}(\N) & \geq \max_{\{p_i,\sigma_i\}\in E_{\sigma}} I_p\Big(\{p_i, \sigma_i\},\N\Big) \\ 
    & =  \Q(\N) - \min_{\{p_i,\sigma_i\}\in E_{\sigma}} \sum_i p_i I_c(\sigma_i,\N)\,,
\end{aligned}
\end{equation*}
it follows that if there exists an ensemble $\{p_i,\sigma_i\}\in E_{\sigma}$ such that $\sum_i p_i I_c(\sigma_i,\N) < 0$, then $P^{(1)} > \Q$.

Now let $\rho$ be a pure state, and let $r>0$ be such that $\sigma-r\cdot\rho(\varepsilon)$ is positive semi-definite for sufficiently small $\varepsilon$ (if the kernel of $\sigma$ is contained in the kernel of $\rho(\varepsilon)$ for some $\varepsilon$, then such an $r$ necessarily exists). Now let $\{q_i,|\psi_i\rangle\langle\psi_i|\}$ be the spectral decomposition of $\sigma-r\cdot\rho(\varepsilon)$ with respect to 1-dimensional projectors (pure states), so that $\{(r,\rho(\varepsilon)),(q_i,|\psi_i\rangle\langle\psi_i|)\}$ is an ensemble decomposition of $\sigma$. We then have 
\begin{equation*}
    r I_c\Big(\rho(\varepsilon),\N\Big) + \sum_i q_i I_c(|\psi_i\rangle\langle\psi_i|,\N) = r I_c\Big(\rho(\varepsilon),\N\Big)\, ,
\end{equation*}
and since $\rho(0) = \rho$ is a pure state, the fact that coherent information is vanishing for pure states implies that the function $f(\varepsilon)$ as defined by Eq.~\eqref{FNCXN71} is then given by $f(\varepsilon) = I_c\big(\rho(\varepsilon),\N\big)$. It then follows that a negative value of $f(\varepsilon)$ implies $P^{(1)}(\N)>\Q(\N)$, as desired.
\end{proof}

\begin{rem}\label{rem5}
In Refs.~\cite{Watanabe,IEEEPartialOrder}, it was shown that if $P(\N^c) = 0$, then 
\begin{equation}\label{equalityQP}
    \Q(\N) = P^{(1)}(\N) = Q(\N) = P(\N)\,.
\end{equation}
As such, a positive gap $P^{(1)}(\N)-\Q(\N)>0$ implies $P(\N^c) >0$. It then follows that the hypotheses of Proposition~\ref{prop2} also provide a sufficient condition for positive private capacity. 
\end{rem}

We now provide a simple example illustrating how Proposition~\ref{prop2} combined with Criterion~\ref{cri1} can be used to detect a positive gap between $P^{(1)}$ and $\Q$.

\bx
Let $\N=\mathcal{D}_p$, where $\mathcal{D}_p$ is the qubit depolarizing channel $\mathcal{D}_p(\rho) = (1-p)\rho + p\, \mathds{1}/2$. It is known that the maximally mixed state $\sigma = \mathds{1}/2$ is optimal for $\N$ for all $p \in [0,0.2524]$. Now consider the perturbative state $\rho(\varepsilon)$ given by 
\[
\rho(\varepsilon) = |0\rangle\langle 0| - \varepsilon \sigma_Z\, ,
\]
with $\varepsilon \in [0,1]$ and $A^{(1)} = -\sigma_Z$. In such a case we have $\sigma - r\cdot\rho(\varepsilon) \geq 0$ for all $r \in (0,1/2]$. We then have
\begin{equation*}
    \begin{aligned}
        \tr\Big(\Pi \N(A^{(1)})\Big) & = 0 \,, \\
        \tr\Big(\Pi^c \N^c(A^{(1)})\Big) & = \frac{p(3-2p)}{2-p} \,,
    \end{aligned}
\end{equation*}
where $\Pi$ and $\Pi^c$ denote the orthogonal projectors onto the kernel of $\N(|0\rangle\langle 0|)$ and $\N^c(|0\rangle\langle 0|)$ respectively. Utilizing Criterion~\ref{cri1}, we find that function $f(\varepsilon)$ has a negative value for all $p \in (0,0.2524]$, and hence $P^{(1)}(\N) > \Q(\N)$ by Proposition~\ref{prop2}. We note that by Remark~\ref{rem5} it then follows that $P(\N^c)>0$. $\hfill\triangle$.
\ex

We note that if $\sigma$ is optimal for a channel $\N$ and there exists a state $\rho$ such that $I_c(\rho,\N^c) > 0$ and $\sigma - r\cdot\rho$ is positive semi-definite for some $r>0$, then the same argument as the proof of Proposition~\ref{prop2} yields $\Q(\N) < P^{(1)}(\N)$, which generalizes the condition for positive gap between $P^{(1)}$ and $\Q$ provided in~\cite{Watanabe,Hirche2022}. There are many examples satisfying such conditions, such as Pauli channels~\cite{complementaryQuantum}, and dephrasure channels $\N_{p,q}$ which have positive one-shot quantum capacity for all $p,q \in \mathcal{R}_1$, where 
\[
\mathcal{R}_1=\left\{(p,q)\,\, \Big{|}\,\, p\in [0,1/2], \,\, 0\leq q < \frac{(1-2p)^2}{1+(1-2p)^2}\right\}\, .
\]
In particular, when $(p,q)\in \mathcal{R}_1$, the complementary channel $\N_{p,q}^c$ has positive one-shot quantum capacity, and the optimal state for $\N_{p,q}$ is of full rank. Therefore, for any state $\rho$ such that $I_c(\rho,\N^c_{p,q})>0$, we can always find a small parameter $r$ such that $\sigma-r\cdot \rho>0$. Then it follows that $P^{(1)}(\N_{p,q})-\Q(\N_{p,q})>0$ for all $(p,q)\in \mathcal{R}_1$. This extends a result in Ref.~\cite{DephrasureChannel}, where it was shown such a positive gap exists in the range $p\in[0.08,0.125]$ and $q = 3p$. 

\section{discussion}

In this work, we performed a perturbative analysis of the behavior of coherent information with respect to small perturbation of the input state, and were able to establish 3 general criteria for determining when a state is not a local maximum for coherent information. We then employed such criteria to derive sufficient conditions for super-additivity of one-shot quantum capacity, and also for detecting a positive gap between one-shot private capacity and one-shot quantum capacity. Theoretically speaking, our criteria can also be utilized for detecting superadditivity of one-shot private capacity. In particular, from Eq.~\eqref{privateinformation} it follows that the private information of a quantum channel $\N$ associated with a general ensemble can be written as the difference between the coherent information $I_c(\sum p_i\rho_i, \N)$ and the convex sum $\sum p_i I_c(\rho_i,\N)$. Now given optimal ensembles $\{p_i,\rho_i\}$ and $\{q_j,\omega_j\}$ realizing the one-shot private capacity of quantum channels $\N$ and $\M$ respectively, one may consider the perturbative ensemble $\{p_i q_j, \sigma_{ij}(\varepsilon)\}$ with $\sigma_{ij}(0) = \rho_i\otimes\omega_j$. Now let $f(\varepsilon)$ be the function given by
\[
f(\varepsilon)=I_p(\varepsilon) - I_p(0)\, ,
\]
where   
\begin{equation*}
    I_p(\varepsilon)= I_p\Big(\{p_iq_j, \sigma_{ij}(\varepsilon)\}, \N\otimes\M\Big)\, .
\end{equation*}
We then have
\begin{equation*}
    f'(\varepsilon) = I_c'(\sigma(\varepsilon),\N\otimes\M) - \sum p_iq_j I_c'(\sigma_{ij}(\varepsilon), \N\otimes\M)\,,
\end{equation*}
where $\sigma(\varepsilon) = \sum p_iq_j \sigma_{ij}(\varepsilon)$. Therefore, for various ensembles and channels, we can use our perturbative analysis of coherent information to detect super-additivity of the one-shot private capacity. However, we have yet to find any examples for which our criteria can be applied. On the other hand, for degradable channels, their private capacities are additive, and equal to the one-shot quantum capacity, hence our results also give necessary conditions for degradability of generic channels.

Similar to one-shot quantum and private capacities, the Holevo capacity $C^{(1)}$ bounds how much classical information can be transmitted by a single use of a quantum channel, which is defined by taking a maximum of Holevo information over all ensembles. Given an ensemble $\{p_i,\rho_i\}$ of input states for a quantum channel $\N$, the Holevo information is the real number $\chi(\{p_i,\rho_i\},\N)$ given by
\begin{equation}\label{Holevoinformation}
    \chi\Big(\{p_i,\rho_i\},\N\Big) = S\big(\N(\rho)\big) - \sum_i p_i S\big(\N(\rho_i)\big)\,,
\end{equation}
where $\rho = \sum_i p_i\rho_i$ is the associated mixture of the ensemble. It then follows from Eq.~\eqref{privateinformation}, that one can rewrite the private information as
\begin{equation}\label{IpChi}
    I_p\Big(\{p_i,\rho_i\},\N\Big) = \chi\Big(\{p_i,\rho_i\},\N\Big) - \chi\Big(\{p_i,\rho_i\},\N^c\Big)\,.
\end{equation}
Due to the concavity of entropy, the private information is naturally less than Holevo information, and moreover, since the Holevo capacity is the maximum of Holevo information over all ensembles, we have
\begin{equation*}
    P^{(1)}(\N) \leq C^{(1)}(\N)\,.
\end{equation*}
In particular, if $\{p_i,\rho_i\}$ achieves the one-shot private capacity of $\N$, Eq.~\eqref{IpChi} yields
\begin{equation*}
    \begin{aligned}
        C^{(1)}(\N) & \geq \chi\Big(\{p_i,\rho_i\},\N\Big) \\ 
        & = P^{(1)}(\N) + \chi\Big(\{p_i,\rho_i\},\N^c\Big)\,,
    \end{aligned}
\end{equation*}
thus the gap between $P^{(1)}$ and $C^{(1)}$ is at least $\chi(\{p_i,\rho_i\},\N^c) \geq 0$. According to the equality condition for concavity of entropy, it follows that $\chi(\{p_i,\rho_i\},\N^c) = 0$ if and only if $\N^c(\rho_i) = \N^c(\rho)$ for all $i$. Consequently, the perturbative method may be of limited utility in identifying the gap between the one-shot private capacity and the Holevo capacity, analogous to its failure in detecting the nonadditivity of the minimum output entropy~\cite{MOELA}. %Moreover, as same of $P^{(1)}$, the optimization of $C^{(1)}$ for a generic channel is also non-convex and in fact is NP-hard~\cite{HolevoNPC}.

%Finally, for decades it was believed that Holevo capacity is additive, i.e., $C^{(1)}(\N_1\otimes\N_2) = C^{(1)}(\N_1) + C^{(1)}(\N_2)$ for all quantum channels $\N_1$ and $\N_2$. In particular, additivity of Holevo capacity is known to hold for many quantum channels, such as unital qubit channels~\cite{Additivity}, entanglement-breaking channels~\cite{Shor2002} and depolarizing channels~\cite{King2003}. In Ref.~\cite{Shor2004}, Shor showed that the additivity of Holevo capacity is equivalent to another additivity conjecture corresponding to the minimum output entropy. However, the additivity conjecture of minimum output entropy was disproved by Hastings~\cite{Hastings2009}, which then disproved the additivity conjecture for Holevo capacity by Shor's result. However, Hasting's proof that the additivity conjecture for minimal output entropy is false did not provide a counterexample, and such a counterexample to additivity of Holevo capacity still eludes us. While one may consider apply our perturbative methods theory to detect the superadditivity of Holevo capacity, however, although minimum output entropy is not additive, it's locally additive~\cite{MOELA}. More precisely, if two pure states $\psi_1$ and $\psi_2$ locally minimize the output entropy of two quantum channels $\N_1$ and $\N_2$ respectively, the tensor product state $\psi_1\otimes\psi_2$ also locally minimizes the output entropy of $\N_1\otimes\N_2$, thus the superadditivity of Holevo capacity is a global effect.

\section{Acknowledgments} 
We thank Bikun Li for helpful comments. ZM is supported by the Fundamental Research Funds for the Central Universities, the National Natural Science Foundation of China (no. 12371132). JF is supported by the Hainan University startup fund for the project ``Spacetime from quantum information".
  
\bibliography{reference}

\clearpage
\newpage

\title{Methods}
\author{testing}

\maketitle
\onecolumngrid
%\tableofcontents
\vspace{1cm}

\begin{center}\large \textbf{Criteria for  super-additivity of coherent information} \\
\textbf{--- Supplementary Material ---}\\
\end{center}

\appendix

\section{Perturbation Theory and The Derivatives of Entropy}

In this section, we review some results in perturbation theory in order to give the exact expression of eigenvalues for the perturbative expansion $\rho(\varepsilon)$ considered in this work.

Suppose that $\rho$ is a quantum state, and consider the perturbative expansion $\rho(\varepsilon)$ of $\rho$ given by
\begin{equation}\label{PertEx}
    \rho(\varepsilon) = \rho + \sum_{i = 1}^\infty \varepsilon^i A^{(i)} \,,
\end{equation}
where $\rho(0) = \rho$ and $A^{(i)}$ are traceless Hermitian operators such that $\rho(\varepsilon)$ is a valid state for small enough $\varepsilon$. Let $\{\lambda_i\}$ be the set of disctict eigenvalues of $\rho$, so that $\lambda_i\neq \lambda_j$ for $i\neq j$, and let $n_i$ denote the multiplicity of $\lambda_i$. Let $P_i$ denote the orthogonal proection onto the $\lambda_i$-eigenspace of $\rho$, so that $\rho = \sum_i \lambda_i P_i$ and $\sum_i P_i = \mathds{1}$. 

We now consider one eigenvalue $\lambda \in \{\lambda_i\}$ with multiplicity $n$, as well as the associated projection $P$ (the index is ignored for simplicity). We then define the operator $S$ given by
\[
S = \sum_{\lambda'\neq \lambda} (\lambda' - \lambda)^{-1} P'\, ,
\]
where $\lambda'$ are other eigenvalues of $\rho$ differing from $\lambda$ and $P'$ is the eigenprojection corresponding to $\lambda'$-eigenspace of $\rho$. It's clear that $S$ satisfies $SP = PS = 0$.

The eigenvalue perturbation theory~\cite{Kato,Baumgartel1984} shows that the eigenvalue $\lambda$ of $\rho$ will in general split into several perturbative eigenvalues $\lambda_j(\varepsilon)$ of $\rho(\varepsilon)$ for small $\varepsilon \neq 0$ as 
\begin{equation*}
     \lambda_{j}(\varepsilon) = \lambda + \lambda_{j1}\varepsilon + \lambda_{j2}\varepsilon^2 +\cdots \qquad j = 1,\cdots,n\, .
\end{equation*}
The exact expression of $\lambda_j(\varepsilon)$ is in general difficult to obtain. Here we recall some properties about the sum of these perturbative eigenvalues.

First, Ref.~\cite{Kato} gives the eigen-projection corresponding to these perturbative eigenvalues as $P(\varepsilon)$.
% \begin{equation*}
%     P(\varepsilon) = P + \sum_{i = 1}^\infty \varepsilon^i P^{(i)}\,. 
% \end{equation*}
% The full expression of $P^{(n)}$ can be found in Eq.~(2.12) of~\cite{Kato}, here we give the first two terms for ease of subsequent discussion:
% \begin{equation*}
%     \begin{aligned}
%         P^{(1)} & = -PA^{(1)}S - SA^{(1)}P \,, \\
%         P^{(2)} & = -PA^{(2)} S- SA^{(2)}P + PA^{(1)} SA^{(1)} S \\ 
%         & \quad + SA^{(1)} PA^{(1)}S + S A^{(1)}SA^{(1)} P - P A^{(1)} P A^{(1)} S^2  \\
%         & \quad - P A^{(1)} S^2 A^{(1)} P - S^2 A^{(1)} P A^{(1)} P \,.
%     \end{aligned}
% \end{equation*}
On the other hand, the weighted mean of these perturbative eigenvalues is given by
\begin{equation*}
    \hat{\lambda}(\varepsilon) = \frac{1}{n} \tr\Big(\rho(\varepsilon)P(\varepsilon)\Big) = \lambda + \frac{1}{n} \tr\Big((\rho(\varepsilon) - \lambda) P(\varepsilon)\Big) \,.
\end{equation*}
It is worth noting that if there is no splitting of $\lambda$, i.e., if these perturbative eigenvalues are all same and equal to an identical $\lambda(\varepsilon)$, then $\hat{\lambda}(\varepsilon) = \lambda(\varepsilon)$. In particular, this equality always holds if $n=1$.

Second, Ref.~\cite{Kato} shows $(\rho(\varepsilon) - \lambda)P(\varepsilon) = \sum_{i=1}^\infty \varepsilon^i T^{(i)}$, and the first two term $T^{(1)}$ and $T^{(2)}$ are 
\begin{equation}\label{OperatorT}
    \begin{aligned}
        T^{(1)} = PA^{(1)}P \,, \qquad \text{and} \qquad
        T^{(2)} = P A^{(2)} P-P A^{(1)} P A^{(1)} S - P A^{(1)}S A^{(1)} P-S A^{(1)} P A^{(1)} P \,.
    \end{aligned}
\end{equation}
The weighted mean $\hat{\lambda}(\varepsilon)$ can then be written as
\begin{equation*}
    \hat{\lambda}(\varepsilon) = \lambda + \sum_{i=1}^\infty \varepsilon^i \lambda^{(i)}\,,\quad \lambda^{(i)} = \frac{1}{n}\tr\big(T^{(i)}\big)\,, \, i \geq 1 \, .
\end{equation*}
In particular, since $\hat{\lambda}(\varepsilon) = \frac{1}{n} \sum_j \lambda_j(\varepsilon)$, we have 
\begin{equation}\label{sumlam}
    \sum_j \lambda_{j1} = \tr(T^{(1)}) = \tr(PA^{(1)})\,,\quad \sum_j \lambda_{j2} = \tr(T^{(2)}) = \tr(PA^{(2)}) - \tr(PA^{(1)}SA^{(1)}P)\,.
\end{equation}
Our later discussion is highly dependent on the equality~\eqref{sumlam} above.

\vspace*{2em}

Now, we apply above discussion about the perturbative eigenvalues to analyze the derivatives of the entropy function $S(\varepsilon)=S(\rho(\varepsilon))$ in a small neighborhood of $\varepsilon=0$. So consider the function $g(\varepsilon) = \lambda(\varepsilon) \log \lambda(\varepsilon)$ with positive $\lambda(\varepsilon) = \lambda_0+\lambda_{1} \varepsilon+\lambda_{2} \varepsilon^2+\cdots$ and $\varepsilon \in [0,R)$ for some radius of convergence $R$. As shown in~\cite{DetectingPositive_NPJ}, there exists $r \in (0,R)$ such that $g(\varepsilon)$ is a real analytic function on $(0,r)$, and moreover, its derivatives 
\begin{equation}\label{DerofG}
    \begin{aligned}
        g'(\varepsilon)  = \lambda_1\log\lambda(\varepsilon) + \lambda_1 + \mathcal{O}(\varepsilon)\,, \qquad
        g''(\varepsilon)  = 2\lambda_2\log\lambda(\varepsilon)+2\lambda_2+\lambda_1^2/\lambda(\varepsilon) + \mathcal{O}(\varepsilon)
    \end{aligned}
\end{equation}
are also analytic on $(0,r)$ (and in particular are continuous), where $\mathcal{O}(\varepsilon) \rightarrow 0$ as $\varepsilon\to 0$.

In general, if we fix an eigenvalue $\lambda_i$ of $\rho$ with multiplicity $n_i$, it will split into $n_i$ perturbative  eigenvalues of $\rho(\varepsilon)$ of the form
\begin{equation*}
    \lambda_{ij}(\varepsilon) = \lambda_i + \lambda_{ij1}\varepsilon + \lambda_{ij2}\varepsilon^2 +\cdots \qquad j = 1,\ldots,n_i\, ,
\end{equation*}
where $\lambda_{ijk}$ corresponds to the operator $T^{(k)}$ for all $k\geq 1$. Let $S(\varepsilon) = S(\rho(\varepsilon))$ be the entropy of $\rho(\varepsilon)$ as a function of $\epsilon$. The first and second derivatives of $S(\varepsilon)$ are then given by
\begin{equation*}
\begin{aligned}
     S'(\varepsilon) & = -\sum_i\sum_j \Big(\lambda_{ij1} + 2\lambda_{ij2}\varepsilon + \cdots\Big)\Big(\log\lambda_{ij}(\varepsilon) + 1\Big) \\
     S''(\varepsilon) & = -\sum_i\sum_j \bigg[\Big(2\lambda_{ij2}+6\lambda_{ij3}\varepsilon+\cdots\Big)\Big(\log\lambda_{ij}(\varepsilon)+1\Big) + \Big(\lambda_{ij1} + 2\lambda_{ij2}\varepsilon + \cdots\Big)^2/\lambda_{ij}(\varepsilon) \bigg]\,.
\end{aligned}
\end{equation*}

We now consider two cases corresponding to whether or not $\rho$ is of full rank.

\noindent \underline{$\rho$ is not of full rank}: In this case we assume without loss of generality that $\lambda_0=0$. When we take the limit $\varepsilon\to 0$, we will meet the term $\lim_{\varepsilon\rightarrow 0^+} \lambda_{0j1}\log \varepsilon$, and once $\lambda_{0j1}\neq 0$, the derivative $S'(0)$ is infinite. When $\lambda_i \neq 0$, it's clear that the terms $\lambda_{ij1}\log \lambda_{ij}(\varepsilon)$, $\lambda_{ij12}\log \lambda_{ij}(\varepsilon)$ and $\lambda_{ij1}^2/\lambda_{ij}(\varepsilon)$ are all bounded when $\varepsilon$ tends to zero. Therefore, the sum for those indices $i$ such that $\lambda_i \neq 0$ is bounded, and when computing the expression of derivatives of entropy, we can ignore the sum of these terms with $\lambda_i \neq 0$, and focus on the sum of terms with $\lambda_0 = 0$, so that
 \begin{equation*}
 \begin{aligned}
      S'(\varepsilon) & = -\sum_j  \Big(\lambda_{0j1} + 2\lambda_{0j2}\varepsilon + \cdots\Big)\Big(\log\lambda_{0j}(\varepsilon)+1\Big) + C  \\
      & =  -\sum_j  \Big(\lambda_{0j1} + 2\lambda_{0j2}\varepsilon + \cdots\Big)\Big(1 + \log \varepsilon + \log (\lambda_{0j1} + \lambda_{0j2}\varepsilon +\cdots)\Big) + C \\
      & =  -\sum_j  \Big(\lambda_{0j1} + 2\lambda_{0j2}\varepsilon + \cdots\Big) \log \varepsilon + C'
 \end{aligned}
 \end{equation*}
 where $C$ and $C'$ are finite real numbers. It is worth mentioning that a condition: $P_0 A^{(1)} P_0 \geq 0$ is contained here in order for $\lambda_{0j1}\log \lambda_{0j1}$ to be meaningful where $P_0$ is the eigen-projection associated with the $\lambda_0$-eigenspace. Now let $\Pi=P_0$ denote the projector onto the kernel of $\rho$. Since $\lim_{x\rightarrow 0^+} x\log x=0$, we have 
 \[
 S'(\varepsilon)\sim -\sum_j \lambda_{0j1} \log\varepsilon = -\text{Tr}(\Pi \, A^{(1)})\log\varepsilon \,,
 \]
where for two real-valued functions $f(\varepsilon)$ and $g(\varepsilon)$, we write $f(\varepsilon) \sim g(\varepsilon) \iff \lim_{\varepsilon\rightarrow 0^+} f(\varepsilon)/g(\varepsilon) = 1$. Similarly, for $S''(\varepsilon)$ we have 
 \begin{equation*}
     \begin{aligned}
         S''(\varepsilon) \sim -\sum_j \Big(2\lambda_{0j2} \log\varepsilon + \lambda_{0j1}^2/\lambda_{0j}(\varepsilon) \Big) \sim -\sum_j \lambda_{0j1}/\varepsilon = -\text{Tr}(\Pi \, A^{(1)})/\varepsilon\,,
     \end{aligned}
 \end{equation*}
where the second $\sim$ follows from the fact that $\lambda_{0j1}^2/\lambda_{0j}(\varepsilon) \sim \lambda_{0j1}/\varepsilon$ and $\lim_{\varepsilon\rightarrow 0^+} \varepsilon\log\varepsilon = 0$.

\noindent \underline{$\rho$ is of full rank}: In this case the distinct eigenvalues $\lambda_i$ of $\rho$ are all positive. We then have
 \begin{align*}
     \lim_{\varepsilon\rightarrow 0^+} S'(\varepsilon) & = -\sum_i\sum_j \lambda_{ij1}\Big(\log\lambda_i + 1\Big) \\
     & = -\sum_i \log\lambda_i \sum_j\lambda_{ij1} \\
     & = -\sum_i \log\lambda_i\ \text{Tr}(P_iA^{(1)}P_i) \\
     & =  - \tr (A^{(1)} \log\rho)\,, \\
\end{align*}
where the second equality follows from Eq.~\eqref{sumlam} together with the fact that $A^{(1)}$ is traceless and $\sum_{ij}\lambda_{ij1} = \sum_{i} \tr(P_iA^{(1)}) = \tr(A^{(1)}) = 0$. Now for every $i$, we let $C_i$ denote the operator given by $C_i = \sum_{j\neq i} (\lambda_j-\lambda_i)^{-1} P_j$. We then have
\begin{align*}
\lim_{\varepsilon\rightarrow 0^+} S''(\varepsilon) & = -\sum_i\sum_j \bigg[2\lambda_{ij2}\Big(\log\lambda_i+1\Big) + \lambda_{ij1}^2/\lambda_i\bigg] \\
& = -\sum_i 2\log\lambda_i \bigg[ \tr \Big(P_i A^{(2)} - P_i A^{(1)} C_i A^{(1)} \Big)\bigg] - \sum_i \tr\Big(P_i A^{(1)}P_i A^{(1)}\Big)/\lambda_i \\
& = -2\tr\Big(A^{(2)} \log \rho\Big) + \sum_i 2\log\lambda_i \tr\Big(P_i A^{(1)}C_i A^{(1)}\Big) - \sum_i \tr\Big(P_i A^{(1)}P_i A^{(1)}\Big)/\lambda_i \,.
 \end{align*}
where the second equality follows from the fact that $\sum_{ij}\lambda_{ij2} = 0$. To see this, note that by Eq.~\eqref{sumlam}, we have 
\begin{equation*}
    \begin{aligned}
        \sum_{ij}\lambda_{ij2} & = \sum_i \Big(\tr(P_i A^{(2)}) - \tr(P_i A^{(1)}C_iA^{(1)}\Big) \\
        & = - \sum_i \tr\Big(P_i A^{(1)}C_iA^{(1)}\Big) \\
        & = -\sum_i \sum_{j\neq i} \tr\Big(P_i A^{(1)} P_j A^{(1)}\Big) \Big/ (\lambda_j - \lambda_i) = 0\,,
    \end{aligned}
\end{equation*}
the second equality follows from the fact that $A^{(2)}$ is traceless and $\sum_i \tr(P_i A^{(2)}) = \tr(A^{(2)}) = 0$. For the the third equality, we use the definition of $C_i$; and for $\tr(P_i A^{(1)} P_j A^{(1)}) /(\lambda_j - \lambda_i)$ every pair of different indices $(i,j)$, we sum $\tr(P_i A^{(1)} P_j A^{(1)}) /(\lambda_j - \lambda_i)$ and $\tr(P_j A^{(1)} P_i A^{(1)}) /(\lambda_i - \lambda_j)$ which are opposite. Hence we have $\sum_{ij} \lambda_{ij2} = 0$.

\vspace*{2em}

In summary, we have proved the following:

\begin{prop}\label{DerEntropy}
Let $S(\varepsilon): = S(\rho(\varepsilon))$ be the entropy of $\rho(\varepsilon)$. Then the following statements hold.
\begin{enumerate}
\item
If $\rho$ is not of full rank, then 
\begin{align*}
S'(\varepsilon)&\sim -\tr(\Pi A^{(1)})\log\varepsilon \\
S''(\varepsilon)&\sim -\tr(\Pi A^{(1)})/\varepsilon\, ,
\end{align*}
where $\Pi$ is the orthogonal projector onto the kernel of $\rho$.
\item
If $\rho$ is of full rank, then 
\begin{align*}
\lim_{\varepsilon\rightarrow 0^+} S'(\varepsilon)  &= - \tr (A^{(1)} \log\rho) \\
\lim_{\varepsilon\rightarrow 0^+} S''(\varepsilon)  &= -\tr(W_{\mathcal{I}}) -2\tr\big(A^{(2)} \log \rho\big) \,,
\end{align*}
where $\mathcal{I}$ denotes the identity channel and $W_{\mathcal{I}}$ is the operator given by \eqref{W}.
\end{enumerate}
\end{prop}

% There is a special situation meeting both two cases in above Proposition, that is for every projectors $P_i$ onto the $\lambda_i$-eigenspace, the traceless Hermitian operator $A^{(1)}$ satisfying $\tr \big(P_i A^{(1)}\big) = 0$, which implies $S'(\varepsilon) = 0$ and the second derivative also follows the Proposition~\ref{DerEntropy}.

To conclude this section, we introduce some notation which is needed not only for the definition of the operator $W_{\mathcal{I}}$ appearing in Proposition~\ref{DerEntropy}, but also for the three criteria in the main text. Given a quantum channel $\M$ and a quantum state $\rho$, let $\{\lambda_i\}$ denote the set of distinct eigenvalues of $\M(\rho)$, and let $P_i$ denote the projector onto the $\lambda_i$-eigenspace of $\M(\rho)$ for all $i$, so that $\M(\rho) = \sum_i \lambda_i P_i$. 

\vspace*{2em}

\noindent \underline{$\M(\rho)$ is not of full rank}: Without loss of generality assume $\lambda_0 = 0$, so that $P_0 = \Pi$ is the orthogonal projector onto the kernel of $\M(\rho)$. We then let $W_\M^{(0)}$ denote the operator given by 
\begin{equation}\label{W0}
    W_\M^{(0)} = P_0 \M(A^{(1)})C_0 \,\M(A^{(1)})P_0\,, \quad C_0 = \sum_{j\neq 0} (\lambda_j-\lambda_0)^{-1} P_j \, .
\end{equation} 
As discussed above, when we consider the output state of $\M$ corresponding to a perturbative expansion $\rho(\varepsilon)$ as the input state, the eigenvalue $\lambda_0 = 0$ of $\M(\rho)$ will split into the perturbative eigenvalues of $\M(\rho(\varepsilon))$ as
\begin{equation*}
    \lambda_{0j}(\varepsilon) = \lambda_{0j1} \varepsilon + \lambda_{0j2}\varepsilon^2 + \cdots\,,
\end{equation*}
for $j = 1,\cdots,n_0$, and we have $\sum_j \lambda_{0j2} = \tr(\Pi A^{(2)} - W^{(0)}_\M)$.

\noindent \underline{$\M(\rho)$ is of full rank}: In this case we have $\lambda_i > 0$ for all $i$. We then then let $W_\M$ denote the operator given by
\begin{equation}\label{W}
    W_{\M} = \sum_i P_i \M(A^{(1)})V_i\,\M(A^{(1)})P_i \, ,\quad V_i=\frac{P_i}{\lambda_i}-2(\log \lambda_i)\cdot\left(\sum_{j\neq i} (\lambda_j-\lambda_i)^{-1}P_j\right)\, .
\end{equation}

%%%%%%%%%%%%%%%%%%%%%%%%%%%%%%%%%%%%%%%%%%%%%%
\section{The proof of Criterion~\ref{cri1}}
%%%%%%%%%%%%%%%%%%%%%%%%%%%%%%%%%%%%%%%%%%%%%%

Let $\rho$ be an initial state for a quantum channel $\N$, let $\rho(\varepsilon)$ be the perturbative expansion of $\rho$ as defined in~\eqref{PertEx}, and let $f(\varepsilon)$ be the function given by
\begin{equation}\label{function}
    f(\varepsilon) = I_c\left(\rho(\varepsilon), \N\right) - I_c(\rho, \N)\,.
\end{equation}
Since $f(0) = 0$ and $f'(\varepsilon) = S'(\N(\rho(\varepsilon))) - S'(\N^c(\rho(\varepsilon)))$ (which is continuous), once the value of $f'(0)$ is positive (or infinity), we can find an interval $(0,\eta)$ such that $f'(\varepsilon) > 0$ for all $\varepsilon \in (0,\eta)$, implying a positive value of $f(\varepsilon)$. We can then use Proposition~\ref{DerEntropy} regarding the derivatives of entropy to compare the first order derivatives of entropy for $\N(\rho(\varepsilon))$ and $\N^c(\rho(\varepsilon))$:
\begin{thm}\label{thm1}
    The following statements hold.
    \begin{enumerate}
        \item[(1).] Suppose $\N(\rho)$ and $\N^c(\rho)$ are not both of full rank, and suppose
        \begin{equation}\label{notfull}
            \tr\Big(\Pi [\N(A^{(1)})]\Big) > \tr\Big(\Pi^c[\N^c(A^{(1)})]\Big)\,,
        \end{equation}
        where $\Pi$ and $\Pi^c$ denote the orthogonal projectors onto $\emph{Ker}(\N(\rho))$ and $\emph{Ker}(\N^c(\rho))$ respectively. Then there exists a positive value of $f(\varepsilon)$.
        \item[(2).] Suppose $\N(\rho)$ and $\N^c(\rho)$ are both of full rank, and suppose 
        \begin{equation}\label{full}
            \tr\Big(\N^c(A^{(1)})\log \N^c(\rho)\Big) > \tr\Big(\N(A^{(1)})\log \N(\rho)\Big)\,.
        \end{equation}
        Then there exists a positive value of $f(\varepsilon)$. 
    \end{enumerate}
\end{thm}

\begin{proof}
Notice that $f'(\varepsilon) = S'(\N(\rho(\varepsilon))) - S'(\N^c(\rho(\varepsilon)))$ and $\N(\rho(\varepsilon)) = \N(\rho) + \sum_{i=1}^\infty \varepsilon^i \N(A^{(i)})$.   
\begin{itemize}
    \item[(1)] If $\N(\rho)$ and $\N^c(\rho)$ are both of not full rank, then let $\Pi$ and $\Pi^c$ denote the orthogonal projectors onto $\text{Ker}(\N(\rho))$ and $\text{Ker}(\N^c(\rho))$ respectively. According to Proposition~\ref{DerEntropy}, the derivatives $S'(\N(\rho(\varepsilon))) \sim -\tr\big(\Pi \N(A^{(1)}) \big)\log \varepsilon$ and $S'(\N^c(\rho(\varepsilon))) \sim -\tr\big(\Pi^c \N^c(A^{(1)})\big) \log \varepsilon$. Therefore, in order to make sure $\lim_{\varepsilon\rightarrow 0^+} f'(\varepsilon) > 0$, we only need 
    \[
    \tr\big(\Pi\, \N(A^{(1)}) \big) > \tr\big(\Pi^c \, \N^c(A^{(1)}) \big)\, .
    \]
    \item[(2)] If one of $\N(\rho)$ and $\N^c(\rho)$ is of full rank, and another is not of full rank. Without loss of generality, suppose that $\N(\rho)$ is of full rank and $\N^c(\rho)$ is not of full rank. Let $\Pi^c$ denote the orthogonal projectors $\emph{Ker}(\N^c(\rho))$, according to Proposition~\ref{DerEntropy}, we can know that 
    $\lim_{\varepsilon\rightarrow 0^+} S'(\N(\rho(\varepsilon))) = - \tr(\N(A^{(1)}) \, \log\N(\rho))$ which is bounded. But the derivative $S'(\N^c(\rho(\varepsilon))) \sim -\tr\big(\Pi^c \N^c(A^{(1)})\big) \log \varepsilon$ and be can tend to infinite when $\tr\big(\Pi^c \N^c(A^{(1)})\big) < 0$. While the orthogonal projectors onto $\emph{Ker}(\N(\rho))$ is $\Pi = 0$, and in such a case, $\tr\big(\Pi\, \N(A^{(1)}) \big) = 0$ always holds. Totally, in order to make sure $\lim_{\varepsilon\rightarrow 0^+} f'(\varepsilon) > 0$, we only need 
    \[
    \tr\big(\Pi\, \N(A^{(1)}) \big) > \tr\big(\Pi^c \, \N^c(A^{(1)}) \big)\, ,
    \]
    which is same as the first case.
    \item[(3)]  If $\N(\rho)$ and $\N^c(\rho)$ are both of full rank, then according to Proposition~\ref{DerEntropy}, we have
    \[
    \lim_{\varepsilon\rightarrow 0^+} f'(\varepsilon) = - \tr\Big(\N(A^{(1)})\log \N(\rho)\Big) + \tr\Big(\N^c(A^{(1)})\log \N^c(\rho)\Big) \,,
    \]
    thus in order to make sure $\lim_{\varepsilon\rightarrow 0^+} f'(\varepsilon) > 0$, we only need Eq.~\eqref{full}.
\end{itemize}
\end{proof}

The first part of Theorem~\ref{thm1} is our Criterion~\ref{cri1} in the main text. However, the following Proposition shows that the second part in Theorem~\ref{thm1} is highly restrictive and can not be used to determine the super-additivity of one-shot quantum capacity. The key of the proof below is that when $\rho_i$ achieves the one-shot quantum capacity of $\N_i$, the functions $f_i(\varepsilon) := I_c(\rho_i(\varepsilon),\N_i) - I_c(\rho_i,\N_i)$ has non-positive derivative $f_i'(0) \leq 0$, otherwise, the state $\rho_i(\varepsilon)$ has higher coherent information than $\rho_i$, which violating the fact that $\rho_i$ are optimal to achieve the one-shot quantum capacity of $\N_i$.

\begin{prop}\label{prop1}
Let $\rho_1$ and $\rho_2$ be optimal states for quantum channels $\N_1$ and $\N_2$ respectively, let $\N=\N_1\otimes \N_2$, let $\rho=\rho_1\otimes \rho_2$, and let $f$ be the function as defined in \eqref{function}. If $\N(\rho)$ and $\N^c(\rho)$ are both of full rank, then for any bipartite traceless Hermitian operator $A^{(1)}$ associated with the perturbative expansion $\rho(\varepsilon)$ as defined in~\eqref{PertEx}\, $f'(0) \leq 0$.
\end{prop}
\begin{proof}
In such a case we have $\rho(\varepsilon) = \rho_1\otimes\rho_2 + \sum_{i=1}^\infty \varepsilon^i A^{(i)}$, where $A^{(i)}$ are bipartite traceless Hermitian operators. For all $i$, let $A^{(1)}_k$ be the partial trace of the traceless operator $A^{(i)}$ onto the $k$th system for $k=1,2$, which is also traceless. Since $\N_k(\rho_k)$ and $\N_k^c(\rho_k)$ are both full rank for both $k = 1,2$, then by Theorem~\ref{thm1} we have
\begin{equation*}
        \tr \Big(\N^c_k(A^{(1)}_k)\log \N^c_k(\rho_k)\Big) \leq \tr \Big(\N_k(A^{(1)}_k)\log \N_k(\rho_k)\Big)\, .
\end{equation*}

Now consider the Bloch representation of $A^{(1)} = \sum_{i=0}^{d_1^2-1}\sum_{j=0}^{d_2^2-1} r_{ij} U_i^{(1)} \otimes U_j^{(2)}.$, where $\{U_0^{(k)} = \mathds{1}/d_k,U_i^{(k)}\}_{i=1}^{d^2_k-1}$ is a basis for the space consisting of all bipartite Hermitian operators, so that $d_k$ is the dimension of Hilbert space $\mathcal{H}_k$ for $k = 1,2$. We then have
    \begin{align*}
        & \tr \Big[\N^c_1\otimes\N_2^c(A^{(1)})\log \Big(\N^c_1(\rho_1)\otimes\N_2^c(\rho_2)\Big)\Big] \\
            & \quad= \sum_{i=0}^{d_1^2-1}\sum_{j=0}^{d_2^2-1} r_{ij} \tr \bigg[\N^c_1(U_i^{(1)})\otimes \N^c_2(U_j^{(2)}) \log\Big(\N^c_1(\rho_1)\otimes\N_2^c(\rho_2) \Big)\bigg] \\
            & \quad = \tr \Big(\N_1^c(A^{(1)}_1)\log \N_1^c(\rho_1)\Big) + \tr \Big(\N_2^c(A^{(2)}_2)\log\N_2^c(\rho_2)\Big)\,,
        \end{align*}
where $A^{(1)}_1 =\sum_{i=0}^{d_1^2-1} r_{i0} U_i^{(1)},\, A^{(1)}_2 = \sum_{j=0}^{d_2^2-1} r_{0j} U_j^{(2)}$ are the reduced operators of $A^{(1)}$. The second equality holds due to the fact that $\log (A\otimes B) = \log A\otimes \mathds{1} + \mathds{1}\otimes \log B$ and $\tr\ \N_1^c(U_i^{(1)}) = \tr\ \N_2^c(U_j^{(2)}) = 0$ for $i,j\geq 1$ since any quantum channel is trace preserving. Therefore,
    \begin{equation*}
        \begin{aligned}
            & \tr \Big[\N^c_1\otimes\N_2^c(A^{(1)})\log\Big(\N^c_1\otimes\N_2^c(\rho_1\otimes \rho_2)\Big)\Big] = \sum_{i=1}^2 \tr \Big(\N^c_i(A^{(1)}_i)\log \N^c_i(\rho_i)\Big) \\
            & \leq \sum_{i=1}^2 \tr \Big(\N_i(A^{(1)}_i)\log \N_i(\rho_i) \Big) = \tr \Big[\N_1\otimes\N_2(A^{(1)})\log\Big(\N_1\otimes\N_2(\rho_1\otimes \rho_2)\Big)\Big]\,,
        \end{aligned}
    \end{equation*}
    hence $f'(0) \leq 0$ as desired.
\end{proof}

\section{The proof of Criterion~\ref{cri2}}

When $\rho$ is optimal for $\N$, and $\N(\rho)$ and $\N^c(\rho)$ are both of full rank, then by Theorem~\ref{thm1} we know that $f'(0) \leq 0$ for all traceless Hermitian operators $A^{(1)}$. Moreover, when $\rho_k$ is optimal for $\N_k$ for $k=1,2$, and $\N_k(\rho_k)$ and $\N^c(\rho_k)$ are all of full rank, Proposition~\ref{prop1} rules out the possibility of using first-order derivatives of entropy to determine superadditivity of one-shot quantum capacity for $\N_1\otimes\N_2$. As such, we then consider the situation that the first-order derivative $f'(0) = 0$ and the second-order derivative $f''(0)>0$, which also implies a positive value of $f(\varepsilon)$. 

Before considering the second derivative of the function $f$, we first derive a condition on $A^{(1)}$ that implies $f'(0) = 0$. 
\begin{rem}
    When $\N(\rho)$ and $\N^c(\rho)$ are both of full rank, it's clear that $A^{(1)}$ must satisfy the equality
\begin{equation*}
    \tr\Big(\N^c(A^{(1)})\log \N^c(\rho)\Big) = \tr\Big(\N(A^{(1)})\log \N(\rho)\Big)\,,
\end{equation*}
which is equivalent to $f'(0) = 0$. 

When $\N(\rho)$ and $\N^c(\rho)$ are not both of full rank, a simple but effective condition on $A^{(1)}$ is 
\begin{equation}\label{conditionA}
    \tr\Big(P_i\N(A^{(1)})\Big) = \tr\Big(Q_s\N^c(A^{(1)})\Big) = 0\,, \quad \text{for}\quad \forall i,s\,,
\end{equation}
where $P_i$ is the projector onto the eigenspace corresponding to the eigenvalue $\lambda_i$ of $\N(\rho)$, and $Q_s$ is the projector onto the eigenspace corresponding to the eigenvalue $\mu_s$ of $\N^c(\rho)$. In fact, according to Eq.~\eqref{OperatorT}~\eqref{sumlam} and \eqref{DerofG}, the perturbative eigenvalues $\lambda_{ij}(\varepsilon)$ and $\mu_{st}(\varepsilon)$ split from the eigenvalue $\lambda_i$ and $\mu_s$ are in the form of: 
$$
\lambda_{ij}(\varepsilon) = \lambda_i + \lambda_{ij1}\varepsilon + \lambda_{ij2}\varepsilon^2 + \cdots \,, \quad \mu_{st}(\varepsilon) = \mu_s + \mu_{st1}\varepsilon + \mu_{st2}\varepsilon^2 + \cdots\,,
$$
We now consider the derivative of $\sum_{ij}\lambda_{ij}(\varepsilon) \log \lambda_{ij}(\varepsilon)$ which is corresponding to the sum $\sum_{ij} \lambda_{ij}'(\varepsilon) \big(\log \lambda_{ij}(\varepsilon) + 1\big)$. While $\sum_{ij} \lambda_{ij1} = \tr(\N(A^{(1)})) = 0$, we only need to compute the sum $\sum_{ij} \lambda_{ij}'(\varepsilon) \log \lambda_{ij}(\varepsilon)$.

Since for all non-zero $\lambda_i$:
$$
\lim_{\varepsilon\rightarrow 0^+}  \sum_j \lambda_{ij}'(\varepsilon) \log \lambda_{ij}(\varepsilon) = \sum_j \lambda_{ij1} \log \lambda_i = \tr \big(P_i \N(A^{(1)}) \big) \log \lambda_i = 0 \, .
$$
When $\lambda_i = 0$, as $\lambda_{ij}(\varepsilon) \geq 0$, we can know that $\lambda_{ij1} \geq 0$. Let $\Pi$ denote the orthogonal projector onto $\emph{Ker}(\N(\rho))$, according to Eq.~\eqref{conditionA}, we have $\sum_j \lambda_{ij1} = \tr (\Pi\,\N(A^{(1)})) = 0$, thus $\lambda_{ij1} = 0$ and
\[
\lim_{\varepsilon\rightarrow 0^+} \sum_j \lambda_{ij}(\varepsilon)'\log\lambda_{ij}(\varepsilon) = \lim_{\varepsilon\rightarrow 0^+} \tr\Big(\Pi\,\N(A^{(1)})\Big) \log \varepsilon + \sum_{j } \lambda_{ij1}\log \lambda_{ij1} = 0\,.
\] 
Therefore, with the condition~\eqref{conditionA}, we can obtain that $\lim_{\varepsilon\rightarrow 0^+} S'(\N(\rho(\varepsilon))) = 0$. Then together with the same argument for $\lim_{\varepsilon\rightarrow 0^+} S'(\N^c(\rho(\varepsilon))) = 0$, we have $\lim_{\varepsilon\rightarrow 0^+} f'(\varepsilon) = 0$.
\end{rem}

Hence with the condition~\eqref{conditionA}, we have $\lim_{\varepsilon\rightarrow 0^+} f(\varepsilon) = \lim_{\varepsilon\rightarrow 0^+} f'(\varepsilon) = 0$. In order to prove Criterion~\ref{cri2} for $\N(\rho)$ and $\N^c(\rho)$ are not both of full rank, we only need to show $\lim_{\varepsilon\rightarrow 0^+} f''(\varepsilon)>0$.
\begin{thm}
    Let $\rho$ be an initial state for a quanrtum channel $\N$, $\rho(\varepsilon) = \rho + \sum_i \varepsilon^i A^{(i)}$ be a perturbative expansion of $\rho$  as defined in~\eqref{PertEx}, and let $f(\varepsilon)$ be the function as defined in~\eqref{function}. Suppose $\N(\rho)$ and $\N^c(\rho)$ are not both of full rank, suppose 
        \begin{equation*}
            \tr\big(P_i \N(A^{(1)}\big) = \tr\big(Q_s\N^c(A^{(1)})\big) = 0 \,, \forall\, i,s
        \end{equation*}
        where $P_i$ are the eigen-projections of $\N(\rho)$ and $Q_s$ are the eigen-projections of $\N^c(\rho)$, and suppose 
        \begin{equation}\label{Cor}
        \tr \left(\Pi^c\, \N^c(A^{(2)}) - W^{(0)}_{\N^c}\right) <  \tr \left(\Pi\, \N(A^{(2)}) - W^{(0)}_{\N}\right)\,,
        \end{equation}        
        where $W^{(0)}_{\N}$ and $W^{(0)}_{\N^c}$ are defined in Eq.~\eqref{W0}, and $\Pi$ and $\Pi^c$ are the projectors onto kernel of $\N(\rho)$ and $\N^c(\rho)$ respectively. Then there exists a positive value of $f(\varepsilon)$.
\end{thm}
\begin{proof}
We only need to consider the second derivative of $f(\varepsilon)$. Using the second derivative of entropy in Proposition~\ref{DerEntropy}, the second-order derivative of $f''(\varepsilon)$ is
\begin{equation*}
        \begin{aligned}
            f''(\varepsilon)  = \sum_{st} \bigg[2\mu_{st2} \log(\mu_{st}(\varepsilon)) + \frac{\mu_{st1}^2}{\mu_{st}(\varepsilon)}\bigg] - \sum_{ij}\bigg[2\lambda_{ij2}\log(\lambda_{ij}(\varepsilon)) + \frac{\lambda_{ij1}^2}{\lambda_{ij}(\varepsilon)}\bigg] + \mathcal{O}(\varepsilon)\,. 
        \end{aligned}
\end{equation*}
Fixing $\mu_s \neq 0$, we have
\[
\lim_{\varepsilon\rightarrow 0^+}\log \mu_{st}(\varepsilon) = \log\mu_s\, ,\quad \text{and} \quad \lim_{\varepsilon\rightarrow 0^+} 1/\mu_{st}(\varepsilon) = 1/\mu_s\,,
\]
thus 
\[
\lim_{\varepsilon\rightarrow0^+} \sum_t \Bigg[2\mu_{st2} \log(\mu_{st}(\varepsilon)) + \frac{\mu_{st1}^2}{\mu_{st}(\varepsilon)}\bigg]
\]
is bounded. Therefore, we can focus on the term corresponding to $\mu_s = 0$ or $\lambda_i = 0$, which may result in the limit being infinite. In particular, we can only consider that both $\N(\rho)$ and $\N^c(\rho)$ are not of full rank, otherwise, for example, if $\N(\rho)$ is of full rank, we then set $\Pi = W_\N^{(0)} = 0$ in following discussion.   

Since $\N(\rho)$ is not of full rank we assume $\lambda_0 = 0$, without loss of generality, which has multiplicity $n_0$, and $P_0 = \Pi$ is the associated projection onto the kernel. Under perturbations, the eigenvalue $\lambda_0$ which split into the perturbative eigenvalues of $\N(\rho(\varepsilon))$ as $\lambda_{0j}(\varepsilon)$ for $j = 1,\cdots,n_0$, as well as let $\mu_0 = 0$ is the eigenvalue of $\N^c(\rho)$ with multiplicity $m_0$ corresponding to $Q_0 = \Pi^c$, and split into the perturbative eigenvalue of $\N^c(\rho(\varepsilon)$ as $\mu_{0t}(\varepsilon)$ for $t=1,\cdots,m_0$. These perturbative eigenvalues $\lambda_{0j}(\varepsilon)$ and $\mu_{0t}(\varepsilon)$ may then be expanded as 
\begin{equation*}
    \lambda_{0j}(\varepsilon) = \lambda_{0j1}\varepsilon + \lambda_{0j2} \varepsilon^2 + \cdots \,,\quad \mu_{0t}(\varepsilon) = \mu_{0t1}\varepsilon + \mu_{0t2}\varepsilon^2 + \cdots\,
\end{equation*}
for $j = 1,\cdots,n_0$ and $t=1,\cdots,m_0$. We now consider the limit $\lim_{\varepsilon\rightarrow 0^+} \sum_{j}\bigg[2\lambda_{0j2}\log(\lambda_{0j}(\varepsilon)) + \frac{\lambda_{0j1}^2}{\lambda_{0j}(\varepsilon)}\bigg]$.

First, as $\tr(\Pi\, \N(A^{(1)})) = 0$, we find that
\begin{equation*}
    \lim_{\varepsilon \rightarrow 0^+} \sum_j \frac{\lambda_{0j1}^2}{\lambda_{0j}(\varepsilon)} = \lim_{\varepsilon\rightarrow 0^+} \frac{\tr\big(\Pi\, \N(A^{(1)})\big)}{\varepsilon} = 0\,.
\end{equation*}
As $\sum_j \lambda_{0j2} = \tr\Big(\Pi \N(A^{(2)}) - W^{(0)}_{\N}\Big)$ and $\log \lambda_{0j}(\varepsilon) \sim \log\varepsilon$ whether or not $\lambda_{0j1}$ equals to $0$, we have 
\begin{equation*}
     \sum_{j} \lambda_{0j2} \log \big(\lambda_{0j}(\varepsilon)\big) \sim  \tr\Big(\Pi \N(A^{(2)}) - W^{(0)}_{\N}\Big) \log \varepsilon\,.
\end{equation*}
Therefore, together with the same result for the sum with respect to $\mu_{0t}(\varepsilon)$, we find $ \lim_{\varepsilon \rightarrow 0^+} f''(\varepsilon) > 0$ if Eq.~\eqref{Cor} holds. Since $ \lim_{\varepsilon \rightarrow 0^+} f'(\varepsilon) = 0$ and $f''(\varepsilon)$ are continuous, we can find an interval $(0,\eta_1)$ so that $f''(\varepsilon) > 0$, which implies $f'(\varepsilon)$ is positive at an interval $(0,\eta_2)$. We have then shown there exists a positive value of $f(\varepsilon)$, as desired.
\end{proof}

\section{The proof of Criterion~\ref{cri3}}

Criterion~\ref{cri3} corresponds to the situation when $\N(\rho)$ and $\N^c(\rho)$ are both of full rank. Here we prove a Theorem which states the condition for the second-order derivative $f''(0) > 0$. Then together with the condition on $A^{(1)}$ so that $ \lim_{\varepsilon \rightarrow 0^+} f'(\varepsilon) = 0$, the second part of the following implies the Criterion~\ref{cri3} in the main text.

\begin{thm}\label{thm2}
    Let $\rho$ be an initial state for a quantum channel $\N$, let $\rho(\varepsilon)$ be the perturbative expansion of $\rho$ as defined in~\eqref{PertEx}, and let the $f(\varepsilon)$ be the function defined in~\eqref{function}. Then the following statements hold.
    \begin{enumerate}
        \item[(1).] Suppose $\N(\rho)$ and $\N^c(\rho)$ are not both of full rank, and suppose 
        \begin{equation}\label{notfull2}
        \tr \left(\Pi^c\, \N^c(A^{(1)})\right) >  \tr \left(\Pi\, \N(A^{(1)})\right)\,,
        \end{equation}    
        where $\Pi$ and $\Pi^c$ denote the orthogonal projectors onto $\emph{Ker}(\N(\rho))$ and $\emph{Ker}(\N^c(\rho))$ respectively. Then $f''(0)>0$.
        \item[(2).] If $\N(\rho)$ and $\N^c(\rho)$ are both of full rank, then $f''(0) > 0$ if and only if
        \begin{equation}\label{full2}
            \begin{aligned}
                 2\,\tr\Big(\N^c(A^{(2)}) \log \N^c(\rho)\Big) + \tr(W_{\N^c}) 
                 >  2\,\tr\Big(\N(A^{(2)}) \log \N(\rho)\Big) + \tr(W_{\N})\, ,
            \end{aligned}
        \end{equation}
        where $W_{\N}$ and $W_{\N^c}$ are defined in Eq.~\eqref{W}.
       \end{enumerate}
\end{thm}

\begin{proof}
Notice that $f''(\varepsilon) = S''(\N(\rho(\varepsilon))) - S''(\N^c(\rho(\varepsilon)))$ and $\N(\rho(\varepsilon)) = \N(\rho) + \sum_{i=1}^\infty \varepsilon^i \N(A^{(i)})$.    
\begin{itemize}
    \item[(1)] If $\N(\rho)$ and $\N^c(\rho)$ are both not of full rank, then let $\Pi$ and $\Pi^c$ denote the orthogonal projectors onto $\emph{Ker}(\N(\rho))$ and $\emph{Ker}(\N^c(\rho))$ respectively. According to Proposition~\ref{DerEntropy}, the derivatives $S''(\N(\rho(\varepsilon))) \sim -\tr\big(\Pi \N(A^{(1)}) \big)/\varepsilon$ and $S''(\N^c(\rho(\varepsilon))) \sim -\tr\big(\Pi^c \N^c(A^{(1)})\big)/ \varepsilon$. Therefore, in order to make sure $\lim_{\varepsilon\rightarrow 0^+} f''(\varepsilon) > 0$, we only need 
    \[
    \tr\big(\Pi\, \N(A^{(1)}) \big) < \tr\big(\Pi^c \, \N^c(A^{(1)}) \big)\, .
    \]
    \item[(2)] If one of $\N(\rho)$ and $\N^c(\rho)$ is of full rank, and another is not of full rank. Without loss of generality, suppose that $\N(\rho)$ is of full rank and $\N^c(\rho)$ is not of full rank, and let $\Pi^c$ denote the orthogonal projector onto $\text{Ker}(\N^c(\rho))$. According to Proposition~\ref{DerEntropy}, it follows that 
    $\lim_{\varepsilon\rightarrow 0^+} S''(\N(\rho(\varepsilon))) = -\tr(W_\N) - 2\tr(\N(A^{(2)}) \, \log\N(\rho))$, which is bounded. However, the derivative $S''(\N^c(\rho(\varepsilon))) \sim -\tr\big(\Pi^c \N^c(A^{(1)})\big)/ \varepsilon$, which can tend to infinity when $\tr\big(\Pi^c \N^c(A^{(1)})\big) > 0$. In such a case we have that the projector $\Pi$ onto $\text{Ker}(\N(\rho))$ is $\Pi = 0$, so that $\tr\big(\Pi\, \N(A^{(1)}) \big) = 0$ always holds. In summary, in order to ensure $\lim_{\varepsilon\rightarrow 0^+} f''(\varepsilon) > 0$, we only need 
    \[
    \tr\big(\Pi\, \N(A^{(1)}) \big) < \tr\big(\Pi^c \, \N^c(A^{(1)}) \big)\, ,
    \]
    which is same as the first case.
    \item[(3)]  If $\N(\rho)$ and $\N^c(\rho)$ are both of full rank, then according to Proposition~\ref{DerEntropy}, we have
    \[
    \lim_{\varepsilon\rightarrow 0^+} f''(\varepsilon) = -\tr(W_\N) - 2\,\tr\Big(\N(A^{(2)})\, \log \N(\rho)\Big) + \tr(W_{\N^c}) + 2\,\tr\Big(\N^c(A^{(2)})\,\log \N^c(\rho)\Big) \,,
    \]
    thus in order to ensure $\lim_{\varepsilon\rightarrow 0^+} f''(\varepsilon) > 0$, we only need Eq.~\eqref{full2}.
\end{itemize}
\end{proof}

When $\N(\rho)$, $\N^c(\rho)$ are both of full rank, in order to ensure $\lim_{\varepsilon\rightarrow 0^+} f'(\varepsilon) = 0$, the Hermitian traceless operator $A^{(1)}$ satisfies the equality $\tr\Big(\N^c(A^{(1)}) \log \N^c(\rho)\Big) = \tr\Big(\N(A^{(1)}) \log \N(\rho)\Big)$, we can also choose $A^{(2)}$ to satisfy this equality, so that the inequality~\eqref{full2} comes to $\tr(W_\N^c) > \tr(W_\N)$, which is the Criterion~\ref{cri3} in the main text. 

\end{document}